\documentclass[letterpaper, 10 pt, conference]{ieeeconf}  

\usepackage{amsmath}
\usepackage{amssymb}
\usepackage{amsfonts}
\usepackage{cite} 
\usepackage{graphicx}
\usepackage{siunitx}
\usepackage{subcaption}
\usepackage{xcolor}
\usepackage{bm}
\usepackage{booktabs}
\usepackage{float}

\newtheorem{theorem}{Theorem}

\newtheorem{proposition}[theorem]{Proposition}

\IEEEoverridecommandlockouts                              

\title{\LARGE \bf
Safe Learning-Based Adaptive Augmentation Control for Fixed-Wing UAV under Uncertainty
}

\author{Leon Raguse$^{1}$,  Lennart Kracke$^{1}$, Mayank Shekhar Jha$^{2}$, Johannes Autenrieb$^{1}$, and Mark Spiller$^{1}$
\thanks{$^{1}$Leon Raguse,  Lennart Kracke, Johannes Autenrieb, and Mark Spiller are with the
	German Aerospace Center (DLR), Institute of Flight Systems, 38108 Braunschweig, Germany. email:  {\tt\small \{leon.raguse, lennart.kracke, johannes.autenrieb,
	mark.spiller\}@dlr.de}}%
\thanks{$^{2}$Mayank Shekhar Jha is with CRAN, CNRS, Universite de Lorraine, France. email: 
{\tt\small mayank-shekhar.jha@univ-lorraine.fr}
	}%
\thanks{This work has been submitted to the IEEE for possible publication.
Copyright may be transferred without notice, after which this version may no longer be accessible.}
}

\begin{document}

\maketitle
\thispagestyle{empty}
\pagestyle{empty}

\begin{abstract}
This paper presents a learning-based adaptive augmentation control concept inspired by the adaptation mechanisms of conventional adaptive control, while not being restricted to their specific parametric adaptation structures. In contrast to augmenting a reinforcement learning (RL) baseline controller with classical adaptive control to account for the simulation-to-reality gap, the proposed approach uses RL-based adaptive augmentation to address the limitations of conventional adaptive control. Domain randomization combined with observation stacking is employed to train the RL-based augmentation to compensate for matched uncertainties in a fixed-wing aircraft system. To ensure constraint satisfaction during operation, a safety filter is incorporated into the control architecture. Based on the concept of pseudo control hedging (PCH), we propose a modified reference model that avoids undesirable interactions between the RL-based augmentation and the safety filter. To reduce the conservatism of the safety filter, we additionally incorporate a disturbance observer. The proposed approach is evaluated on a fixed-wing aircraft model subject to uncertainties.

\end{abstract}


\section{Introduction}
\label{sec:introdcution}
Driven by advances in deep learning and computational capabilities, reinforcement learning (RL) can be applied to increasingly complex control problems. Through interaction with the environment, RL methods can address nonlinear aircraft dynamics and offer the potential to improve control performance and robustness under model uncertainties and varying operating conditions. This potential has motivated increasing research interest in RL-based control of fixed-wing aircraft \cite{richter2024review,marquis2026adversarial}.

Since aircraft control inputs are continuous, RL-based flight control studies predominantly employ algorithms designed for continuous action spaces. The use of deep neural networks further enables these methods to represent complex nonlinear control policies and value functions. Accordingly, deep actor–critic methods such as DDPG \cite{de2023deep,shukla2024reinforcement}, PPO \cite{chowdhury2024interchangeable,bohn2019deep,marquis2026adversarial,chowdhury2024unified}, and SAC \cite{dally2022soft,bohn2023data} are particularly prevalent.

A different class of approaches employs approximate dynamic programming (ADP) for online adaptive flight control. These methods use Bellman-based optimization to update a control policy from an online identified local model, often with a prescribed quadratic value function structure. This enables online adaptation with relatively low computational and data requirements. However, the prescribed value function structure limits the policy representation, while the adaptation depends on the quality of the online system identification \cite{dias2019intelligent,konatala2024flight,konatala2021reinforcement}.

Domain randomization is widely used to expose policies to variations in system and environmental parameters, thereby promoting robustness to model uncertainty and facilitating transfer from simulation to the real system \cite{wada2022sim}. However, a memoryless domain-randomized policy can potentially become conservative because it must perform across the entire randomized domain, whereas a history-based policy can infer the current domain from recent observations and adapt its behavior accordingly \cite{chen2022understanding,margolis2024rapid}. Such temporal information can be provided through observation stacking or recurrent architectures such as long short-term memories (LSTMs), both of which have been applied to RL-based fixed-wing flight control \cite{chowdhury2024unified,chowdhury2024interchangeable,bohn2023data}. Accommodating a broad range of dynamic conditions, however, increases the exploration and training burden, as the policy must learn effective behavior across substantially different system dynamics.

An alternative to domain randomization is to combine RL with conventional adaptive control. The RL policy defines a nominal closed-loop reference model, while the adaptive controller compensates for discrepancies between the nominal and actual dynamics online \cite{Annaswamy2023RLACTAC,Borghesi2026MRARL,Guha2021MRACRL,KannanACCMRACRL,Cheng2022L1RL}. However, when a nominal model is available, using RL to learn the nominal control policy adds an unnecessary learning problem, since a stabilizing model-based controller can be designed directly. Moreover, the achievable compensation is constrained by the assumptions underlying conventional adaptive control laws such as model reference adaptive control (MRAC) and $\mathcal{L}_1$ adaptive control, for which guarantees typically rely on constant or slowly varying uncertainty parameters, limiting their applicability to strongly time-varying aircraft dynamics \cite{goel2024composite}.

Control barrier function (CBF)-based safety filters allow to enforce safety during RL exploration by modifying potentially unsafe inputs \cite{cheng2019end}. However, safety-filtered inputs introduce an interaction with RL, which must account for the discrepancy between the proposed and applied inputs during learning \cite{cheng2019end,emam2022safe}. Moreover, projection-based filters can map different policy inputs to the same safety-filtered input, reducing distinguishability during learning and potentially degrading policy optimization \cite{Markgraf2026}. Approaches have been proposed to mitigate this interaction by learning previous filter corrections \cite{cheng2019end} or differentiating through the safety layer \cite{emam2022safe}. However, projection-based filters can still introduce information loss and affect policy learning \cite{Markgraf2026}. Alternatively, safety constraints can be incorporated directly into the RL objective through barrier or constraint penalties \cite{marvi2021safe,zhao2023stable}, avoiding modification of the applied input but coupling learning to the imposed constraints.

In this work, we propose an approach that retains the model-based controller as the nominal stabilization law and uses RL exclusively for adaptive augmentation. This avoids RL exploration for baseline controller synthesis while allowing the learned augmentation to capture more complex and strongly time-varying uncertainties without imposing a predefined parametric adaptation structure. Domain randomization is used during training to expose the augmentation policy to the expected range of uncertainties, while observation stacking provides temporal information to distinguish their evolving effects.

Furthermore, we propose a fundamentally new solution to the interaction problem between RL and the safety filter. We make use of the pseudo control hedging (PCH) concept \cite{johnson2000pseudo}, originally designed to account for actuator saturation effects. Instead, we apply PCH to ensure that the RL-based adaptive augmentation does not compensate for interventions introduced by the safety filter. Specifically, this is achieved by modifying the reference model such that the resulting model-matching error dynamics are invariant to the input discrepancy introduced by the safety filter. Consequently, the RL-based augmentation is decoupled from the safety filter when minimizing the model matching error.

\section{Background and Problem Formulation}
\label{sec:background_problem_formulation}
Consider the nonlinear control-affine system 
\begin{equation}\label{eq:nonlinear_control_affine_system}
    \dot x = f(x) + g(x) u 
\end{equation}
with state $x \in \mathcal{X} \subset \mathbb{R}^n$, control input $u \in \mathcal{U} \subset \mathbb{R}^m$ and locally Lipschitz continuous functions $f: \mathcal{X} \rightarrow \mathbb{R}^n$ and $g: \mathcal{X} \rightarrow \mathbb{R}^{n \times m}$.

Assume a 0-super-level set $\mathcal{C} = \{x \in \mathcal{X}: h(x) \geq 0 \}\subset \mathcal{X}$ defined by a $r$-th order continuously differentiable function $h: \mathcal{X} \rightarrow \mathbb{R}$. Consider
\begin{align*}
 \psi_0(x)&=h(x),\nonumber\\
 \psi_i(x)&=
 \dot \psi_{i-1}(x)+\alpha_i(\psi_{i-1}(x)),
             \quad i=1,\ldots,r-1,\nonumber\\
 \psi_r(x,u)&
 =
\dot\psi_{r-1}(x,u)
              +\alpha_r(\psi_{r-1}(x)),
\end{align*}
where $\alpha_i$ are differential class $\mathcal K$ functions. Further, assume the associated sets 
 \(
 \mathcal C_i=\{x\in\mathcal X:\psi_{i-1}(x)\geq0\}
 \),
 \(
 \mathcal C_H=\bigcap_{i=1}^{r}\mathcal C_i
 \). The function $h$ is a higher order control barrier function (HOCBF), if
 \begin{equation}
 \sup_{u\in\mathcal U}\psi_r(x,u)\geq0
 \label{eq:mj-hocbf-feasibility}
\end{equation}
for all $x\in\mathcal C_H$, rendering the set $\mathcal C_H$ forward invariant \cite{xiaoControlBarrierFunctions2019}.
Let $\psi_{j,r_j}(x,u)\geq0$ denote the condition (\ref{eq:mj-hocbf-feasibility}) for a HOCBF $h_j$ with relative degree $r_j$. 
For a possibly unsafe nominal controller $u_{\mathrm{nom}}$ the following quadratic program (QP)  
 \begin{align}
 u(x)&=\mathop{\arg\min}_{u\in\mathcal U}
       \frac12\|u-u_{\rm nom}(x)\|^2
        \label{eq:mj-joint-cbf-qp}\\
 &\text{subject to}\quad
       \psi_{j,r_j}(x,u)\geq0,
       \quad
       j=1,\dots,N_{\mathrm{CBF}},\notag
 \end{align}
can be considered to find a minimum invasive control input $u$ satisfying the condition (\ref{eq:mj-hocbf-feasibility}) for $j=1,\dots,N_{\mathrm{CBF}}$ HOCBFs.

In practical applications, the nominal system (\ref{eq:nonlinear_control_affine_system}) rarely captures the true dynamics. Instead, model uncertainties are present, such that the true dynamics are given by
\begin{align}
    \dot x = f(x) + \Delta f(x) + (g(x) + \Delta g(x)) u,
    \label{eq:true_nonlinear_control_affine_system}
\end{align}
where $\Delta f(x)$ and $\Delta g(x)$ are additive and multiplicative uncertainties, respectively. The uncertainties are assumed to be matched, i.e., their effect lies within the range space of the input matrix \(g(x)\). We formulate the lumped uncertainty $\Delta(x,u) = \Delta f(x) + \Delta g(x) u$ and assume the finite bounds
\begin{align}
    \lVert \Delta(x,u) \rVert \leq \Delta_{\mathrm{max}} \quad \lVert \dot\Delta(x,u) \rVert \leq \dot\Delta_{\mathrm{max}} .
    \label{eq:_bounded_uncertainty_lumped}
\end{align}
\begin{figure*}[t]
    \centering
    \includegraphics[width=0.99\textwidth]{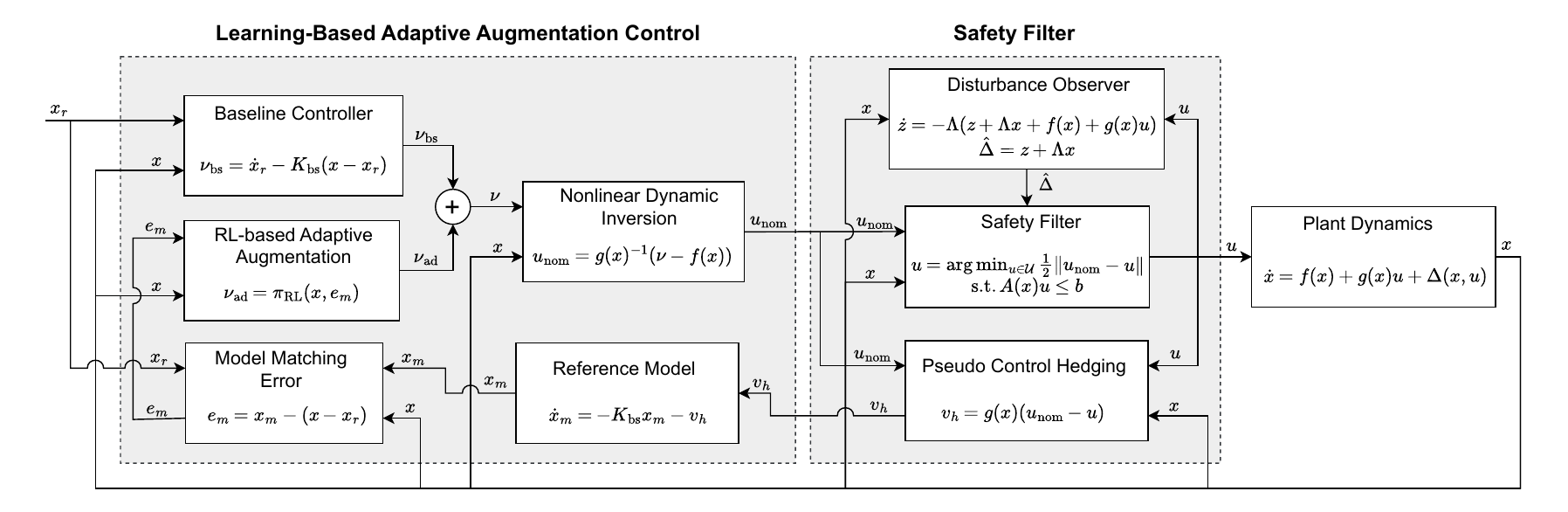}
    \caption{Proposed RL-Based Adaptive Augmentation Control Scheme with Safety Filter}
    \label{fig:block_diagram_overall_concept}
\end{figure*}
For system (\ref{eq:true_nonlinear_control_affine_system}) with lumped uncertainty  $\Delta(x,u)$, we propose the safe RL-based adaptive augmentation control scheme, as illustrated in Fig. \ref{fig:block_diagram_overall_concept}.  Based on the nominal system  (\ref{eq:nonlinear_control_affine_system}), we apply nonlinear dynamic inversion (NDI) and define as pseudo control input $\nu=\nu_{\mathrm{bs}}+\nu_{\mathrm{ad}}$, where $\nu_{\mathrm{bs}}=\dot x_r -K_{\mathrm{bs}}e$ is the baseline controller with feedback matrix $K_{\mathrm{bs}}$,  tracking error $e=x-x_r$ and reference state $x_r$  and $\nu_{\mathrm{ad}}$ is a RL-based adaptive augmentation. 

It follows from Fig. (\ref{fig:block_diagram_overall_concept}) that the closed-loop error dynamics are $\dot e= -K_{\mathrm{bs}}e+\nu_{\mathrm{ad}}+\Delta-v_h$, which are disturbed by the safety filter, i.e., the hedging signal $v_h=g(x)(u_\mathrm{nom}-u)$, that overwrites the nominal inputs $u_{\mathrm{nom}}$ when safety constraints are required to be enforced and by the lumped uncertainty $\Delta(x,u)$.  We use the RL-based augmentation to compensate 
for the uncertainties $\Delta(x,u)$ in the error dynamics. Therefore, we define the reference model based on the desired tracking error dynamics
$\dot x_m =-K_{\mathrm{bs}}x_m-v_h$ and based on that the model matching error $e_m=x_m-(x-x_r)$. Following the design idea of MRAC, we apply the RL-based adaptive augmentation to primarily drive $e_m$ to zero to achieve the reference dynamics. By feeding back the hedging signal $v_h$ into the reference model, the model matching error dynamics  $\dot e_m=-K_{\mathrm{bs}}e_m-\nu_{\mathrm{ad}}-\Delta$ become invariant from the safety filter action. Therefore, the proposed learning-based control scheme provides two main benefits:
\begin{itemize}
    \item Instead of learning the true dynamics, we formulate a RL-based augmentation that aims for learning an adaption law to compensate matched uncertainties similar to MRAC.
    \item We formulate PCH for safety filters. We hide the safety filter action from the model matching error dynamics which allows to train the learning-based control scheme independently of the safety filter. 
\end{itemize}
For the safety filter, we additionally employ a disturbance observer to reduce conservatism with respect to system uncertainties, following the concept of \cite{dacs2022robust}. Based on assumption (\ref{eq:_bounded_uncertainty_lumped}), finite bounds for the disturbance estimation error $\tilde \Delta(x,u) =\Delta(x,u)- \hat \Delta(x,u)$ can be stated.

\section{Aircraft Dynamics}
\label{sec:aircarft_modeling}
The rotational dynamics of the fixed-wing UAV are 
\begin{subequations}
\label{num:rigid_body_dynamics}
	\begin{align}
		\dot \Theta
		&=
		\Psi\omega,\label{num:rigid_body_dynamics_a}\\
		\dot \omega &= 
		-J^{-1}(\omega \times J\omega)
		+
		J^{-1}M,
	\end{align}
\end{subequations}
with the states being: 
the rates around the axis of the body-fixed frame 
\( 
\omega=
\begin{bmatrix}
	p & q & r
\end{bmatrix}^{\top}
\in \mathbb{R}^3
\) 
and the Euler angles
\( 
\Theta=
\begin{bmatrix}
	\phi & \theta & \psi
\end{bmatrix}^{\top}
\in \mathbb{R}^3
\). The inertia tensor is given by $J\in\mathbb{R}^{3 \times 3}$ and the Euler angle rate transformation matrix is denoted by $\Psi\triangleq\Psi(\phi,\theta)\in\mathbb{R}^{3 \times 3}$.

The dynamic pressure is calculated as
\(
Q = 
0.5 \rho_{\mathrm{air}}(h_{\mathrm{alt}}) v_{\mathrm{TAS}}^2 \notag
\), where $\rho_{\mathrm{air}}(h_{\mathrm{alt}})\in\mathbb{R}_{>0}$ is the air density depending on the alitude $h_{\mathrm{alt}}\in\mathbb{R}_{>0}$ and $v_{\mathrm{TAS}}\in\mathbb{R}_{>0}$
is the true airspeed. The matrix $T_{fa}\triangleq T_{fa}(\alpha,\beta)$
is the rotational matrix from the aerodynamic to the body-fixed frame, where $\alpha\in\mathbb{R}$ 
is the angle of attack and
$\beta\in\mathbb{R}$ is the sideslip angle. We denote $S_{\mathrm{ref}}, C_{\mathrm{ref}}, B_{\mathrm{ref}}\in\mathbb{R}_{>0}$ as the reference wing area, the reference chord length and the reference span, respectively, and we define normalized rates as 
\(
\bar p 
= \frac{B_{\mathrm{ref}}}{2 v_{\mathrm{TAS}}}p
\),
\( \bar q = \frac{C_{\mathrm{ref}}}{2 v_{\mathrm{TAS}}}q\), 
\(\bar r = 
	\frac{B_{\mathrm{ref}}}{2 v_{\mathrm{TAS}}}r\).
Additionally, we define
$\xi, \eta, \zeta\in\mathbb{R}$ as the aileron, elevator, and rudder deflections, respectively, and   
\(
\delta
=
\begin{bmatrix}
	\xi & \eta & \zeta
\end{bmatrix}^{\top}
\) as the vector of control surface deflections.

For the moments, we assume the propulsive force to be aligned with the aircraft's centerline. Based on nominal system knowledge, we calculate the nominal moments as
\begin{align}
	M_0 = 
	\Xi(\chi,\omega)
	+
	\Pi(\chi,\omega)
	\delta
	\label{num:aero_moments_explict_expression_nom}	
\end{align}
with
\begin{align}
	\Xi(\chi,\omega)
	&= 
	Q S_{\mathrm{ref}} T_{fa}
	\begin{bmatrix}
		B_{\mathrm{ref}}\bar C_l\\
		C_{\mathrm{ref}}\bar C_m\\
		B_{\mathrm{ref}}\bar C_{n}
	\end{bmatrix},	\notag\\
	\Pi(\chi,\omega)
	&= Q S_{\mathrm{ref}} T_{fa}
		\begin{bmatrix}
			B_{\mathrm{ref}}
			C_{l,\xi} & B_{\mathrm{ref}}
			C_{l,\eta}  & 
			B_{\mathrm{ref}}
			C_{l,\zeta} \\
			C_{\mathrm{ref}}
			C_{m,\xi} & 
			C_{\mathrm{ref}}
			C_{m,\eta} & 
			C_{\mathrm{ref}}
			C_{m,\zeta}\\
			B_{\mathrm{ref}}
			C_{n,\xi} & 
			B_{\mathrm{ref}}
			C_{n,\eta} & 
			B_{\mathrm{ref}}
			C_{n,\zeta} 
		\end{bmatrix}\notag,
\end{align}
where $\chi$ denotes the dependency on the states 
\(
\chi
=
\begin{bmatrix}
\alpha & \beta & v_{\mathrm{TAS}} & h_{\mathrm{alt}}
\end{bmatrix}^{\top}
\).
The aerodynamic moment coefficients of roll, pitch and yaw are defined by functions
$\bar C_l = f_l(\alpha,\beta,\bar p, \bar q, \bar r)$, $\bar C_m = f_m(\alpha,\beta,\bar p, \bar q, \bar r)$, $ \bar C_n = f_n(\alpha,\beta,\bar p, \bar q, \bar r)$ respectively, the effects of the control inputs are described by the linear control derivatives $C_{l,\xi}, C_{l,\eta}, C_{l,\zeta}, C_{m,\xi}, C_{m,\eta}, C_{m,\zeta}, C_{n,\xi}, C_{n,\eta}, C_{n,\zeta}\in\mathbb{R}$.

As the nominal aerodynamic model typically contains significant uncertainty, we define the following uncertainties $\Delta \Xi\triangleq\Delta \Xi(\Delta \bar C_\star) \in\mathbb{R}^3$ and $\Delta \Pi\triangleq\Delta \Pi(\Delta \bar C_{\star,\xi}, \Delta \bar C_{\star,\eta}, \Delta \bar C_{\star,\zeta}) \in\mathbb{R}^{3 \times 3}$ induced by the aerodynamics parametric perturbations $\Delta \bar C_\star, \Delta \bar C_{\star,\xi}, \Delta \bar C_{\star,\eta}, \Delta \bar C_{\star,\zeta}$ with $\star\in\{l,m,n\}$. Consequently, the description of the true moments becomes
\begin{align}
	M = 
	\Xi(\chi,\omega)
	+
	\Pi(\chi,\omega)
	\delta
	+
	\Delta \Xi(\chi,\omega)
    +
	\Delta \Pi(\chi,\omega)
	\delta
	\label{num:aero_moments_explict_expression}.
\end{align}
Let 
\(
x
=
\begin{bmatrix}
\Theta &
\omega &
\chi
\end{bmatrix}^{\top}\in\mathcal{X}\subseteq\mathbb{R}^{10}
\)
and
\(
\delta\in\mathcal{U}\subseteq\mathbb{R}^3
\), then substituting (\ref{num:aero_moments_explict_expression}) in (\ref{num:rigid_body_dynamics}) yields the rotational dynamics
\begin{subequations}
\label{num:rigid_body_dynamics_final}
\begin{align}
\dot \Theta
&=
\Psi\omega,\\
\dot{\omega}
&=
f_{\omega}(x)
+
\Delta f_{\omega}(x)
+
(
g_{\omega}(x)
+
\Delta g_{\omega}(x)
)\delta 
\label{num:rot_dyanmics_perturbated_underbrace}
\end{align}
\end{subequations}
with nominal models $f_{\omega}(x)=-J^{-1}(\omega \times J\omega)
+ J^{-1}\Xi(\chi,\omega)$, $g_{\omega}(x)=J^{-1}\Pi(\chi,\omega)$, additive uncertainty $\Delta f_\omega(x)=J^{-1}\Delta\Xi(\chi,\omega)$ and multiplicative uncertainty $\Delta g_\omega(x)=J^{-1}\Delta\Pi(\chi,\omega)$. 

We assume the nominal control-effectiveness matrix $g_\omega(x)$ to be nonsingular in the considered flight envelope $\mathcal{X}$, implying matched uncertainties. For the lumped uncertainty $\Delta(x,\delta)=\Delta f_\omega(x) + \Delta g_\omega(x)\delta$  we assume
\begin{align}
    \| \Delta(x,\delta) \| \leq \Delta_{\mathrm{max}} \quad \| \dot \Delta(x,\delta)\| \leq \dot\Delta_{\mathrm{max}} .
    \label{eq:_bounded_uncertainty_lumped_aircraft}
\end{align}
As we consider attitude control based on the rotation dynamics (\ref{num:rigid_body_dynamics_final}) only, we also assume that there exists an outer-loop controller that keeps the remaining states in $x$ stable.

\section{Baseline Controller and Reference Model}
\label{sec:baseline_controller}
Consider the nominal control law
\begin{align}
	\delta_{\mathrm{nom}} 
	=
	g_{\omega}(x)^{-1}(\nu 
	-f_{\omega}(x)
	),
	\label{num:NDI_trans_law}
\end{align}
formulated based on the NDI transformation, where $\nu\in\mathbb{R}^3$ denotes the pseudo control input. As we apply the safety filtered input $\delta=\delta_{\mathrm{nom}}-\delta_h$ with $\delta_h = \delta_{\mathrm{nom}}-\delta$ to the system (\ref{num:rot_dyanmics_perturbated_underbrace}) this yields the closed-loop dynamics
\begin{align}
\dot{\omega}
&=
\nu
-
g_{\omega}(x)\delta_h
+
\Delta f_{\omega}(x)
+
\Delta g_{\omega}(x)\delta.
\label{num:true_closed_loop}
\end{align}
Define 
\begin{align}
 g_{\omega}(x)
&\triangleq
\begin{bmatrix}
	 g_{\phi} &
	 g_{\theta} &
	 g_{r}
\end{bmatrix}^{\top},
\label{num:def_elements_g_matrix}\\
\Delta f_{\omega}(x)
&\triangleq
\begin{bmatrix}
	\Delta f_{\phi} &
	\Delta f_{\theta} &
	\Delta f_{r}
\end{bmatrix}^{\top},
\label{num:def_elements_f_uncertainty}\\
\Delta g_{\omega}(x)
&\triangleq
\begin{bmatrix}
	\Delta g_{\phi} &
	\Delta g_{\theta} &
	\Delta g_{r}
\end{bmatrix}^{\top}
\label{num:def_elements_g_uncertainty}
\end{align}
with $\Delta f_{\phi}, \Delta f_{\theta}, \Delta f_{r} \in \mathbb{R}$ and $g_{\phi}, g_{\theta}, g_{r}, \Delta g_{\phi}, \Delta g_{\theta}, \Delta g_{r} \in \mathbb{R}^{1\times3}$, and consider $\nu=\begin{bmatrix}
    \nu_\phi & \nu_\theta & \nu_r
\end{bmatrix}^{\top}$, 
then (\ref{num:true_closed_loop}) can be expressed element-wise as
\begin{align}
\begin{bmatrix}
    \dot p\\
    \dot q\\
    \dot r
\end{bmatrix}
&= 
\begin{bmatrix}
    \nu_\phi\\
    \nu_\theta\\
    \nu_r
\end{bmatrix}
	-
\begin{bmatrix}
    g_{\phi} \delta_h\\
    g_{\theta} \delta_h\\
    g_{r} \delta_h
\end{bmatrix}
	+
\begin{bmatrix}
    \Delta f_{\phi}\\
    \Delta f_{\theta}\\
    \Delta f_{r}
\end{bmatrix}
    +
\begin{bmatrix}
    \Delta g_{\phi} \delta\\
    \Delta g_{\theta} \delta\\
    \Delta g_{r} \delta
\end{bmatrix}.
	\label{num:true_closed_loop_elementwise}	
\end{align}
From (\ref{num:rigid_body_dynamics_a}), the kinematics 
\begin{subequations}\label{num:kinematics}
    \begin{align}
	\ddot\phi	
	&=
	\dot p
	+ \dot q\sin(\phi)\tan(\theta)
	+ \dot r\cos(\phi)\tan(\theta) \notag\\
	& \quad
	+ q\dot \phi\cos(\phi)\tan(\theta)
	+ q\dot \theta\sin(\phi)\sec^2(\theta) \notag\\
	& \quad
	- r\dot\phi\sin(\phi)\tan(\theta)
	+ r\dot\theta\cos(\phi)\sec^2(\theta),
	\label{num:kinematics_phi}\\
	\ddot\theta
	&=
	\dot q\cos(\phi)
	-q\dot\phi\sin(\phi)
	-\dot r\sin(\phi)
	-r\dot \phi\cos(\phi)
	\label{num:kinematics_theta}
	\end{align}
\end{subequations}
are derived. Let the pseudo control inputs be
\begin{subequations}\label{num:pseudo_control_inputs}
\begin{align}
	\nu_{\phi}
	&=
	- \dot q\sin(\phi)\tan(\theta)
	- \dot r\cos(\phi)\tan(\theta)\notag\\
	& \quad
	- q\dot \phi\cos(\phi)\tan(\theta)
	- q\dot \theta \sin(\phi)\sec^2(\theta) \notag\\
	& \quad
	+r\dot\phi\sin(\phi)\tan(\theta)
	- r\dot\theta\cos(\phi)\sec^2(\theta)\notag\\
	& \quad
	+\ddot\phi_r +\mu_{\phi},
	\label{num:nu_phi}\\
	\nu_{\theta}
	&=
	\frac{q\dot\phi\sin(\phi)
	+\dot r\sin(\phi)
	+r\dot \phi\cos(\phi)
	+\ddot\theta_r +\mu_{\theta}}{\cos(\phi)},
	\label{num:nu_theta}\\
	\nu_{r}
	&=\dot r_r +\mu_{r},
	\label{num:nu_r}
\end{align}
\end{subequations}
where $\mu_{\phi}, \mu_{\theta}, \mu_{r}\in \mathbb{R}$ are yet undefined auxiliary control inputs and $\phi_r$, $\theta_r$, $r_r$ are reference signals of the control variables $\phi$, $\theta$, $r$. Substituting the pseudo control inputs (\ref{num:pseudo_control_inputs}) into (\ref{num:true_closed_loop_elementwise}) and then substituting $\dot p$ and $\dot r$ of (\ref{num:true_closed_loop_elementwise}) into (\ref{num:kinematics}) yields
\begin{align}
\begin{bmatrix}
    \ddot \phi\\
    \ddot \theta\\
    \dot r
\end{bmatrix}
&= 
\begin{bmatrix}
    \ddot \phi_r + \mu_\phi 
    + \Delta f_{\phi} + \Delta g_{\phi} \delta
    - g_{\phi} \delta_h\\
    \ddot \theta_r + \mu_\theta 
    +
    \cos(\phi)(\Delta f_{\theta} + \Delta g_{\theta} \delta
    - g_{\theta} \delta_h) \\
    \dot r_r +\mu_{r} 
    + \Delta f_{r} + \Delta g_{r} \delta
    -g_{r} \delta_h
\end{bmatrix}.
\label{num:error_dynamics_baseline_minus1}
\end{align}
By defining the tracking errors $e_\phi=\phi-\phi_r$, $e_\theta=\theta-\theta_r$, $e_r=r-r_r$ and rearranging (\ref{num:error_dynamics_baseline_minus1}) accordingly, we obtain the error dynamics 
\begin{subequations}\label{num:error_dynamics_systems}
\begin{align}
\dot x_\phi
&= 
A_\phi x_\phi
+
B_\phi
(
\mu_{\phi} 
+ \Delta f_{\phi} + \Delta g_{\phi}\delta
-g_{\phi} \delta_h
),\\
\dot x_\theta
&= 
A_\theta x_\theta
+
B_\theta 
(\mu_{\theta}  \notag\\
&  \qquad \qquad 
+
\cos(\phi)
(
\Delta f_{\theta} 
+ \Delta g_{\theta}  \delta 
-g_{\theta}  \delta_h
)
),\\
\dot x_r
&= 
A_r x_r
+
B_r
(    
\mu_{r}
+ \Delta f_{r} + \Delta g_{r} \delta
-g_{r} \delta_h)
\end{align}
\end{subequations}
with
\begin{align}
A_i =
\begin{bmatrix}
    0 & 1 \\
    0 & 0
\end{bmatrix},
\quad
B_i
=
\begin{bmatrix}
    0  \\
    1
\end{bmatrix},
\quad
x_i
=
\begin{bmatrix}
    e_i  \\
    \dot e_i
\end{bmatrix}
\notag
\end{align}
for $i\in\{\phi,\theta\}$ and $A_r=0$, $B_r=1$, $x_r=e_r$. 
We decompose the auxiliary control inputs as $\mu = \mu_\mathrm{bs}+\mu_\mathrm{ad}$, where 
\(
\mu_\mathrm{bs}
=
\begin{bmatrix}
     \mu_{\phi,\mathrm{bs}} &
     \mu_{\theta,\mathrm{bs}} &
     \mu_{r,\mathrm{bs}}  
\end{bmatrix}^{\top} 
\)
denotes the baseline error controller and 
\(
\mu_\mathrm{ad} 
=
\begin{bmatrix}
      \mu_{\phi,\mathrm{ad}}&
      \mu_{\theta,\mathrm{ad}}&
      \mu_{r,\mathrm{ad}}
\end{bmatrix}^{\top}
\)
the adaptive augmentation. Consider $i\in\{\phi,\theta, r\}$, we chose linear baseline controllers 
\(
\mu_{i,\mathrm{bs}} = K_i x_i,
\)
and for the dynamics  
\begin{align}
\dot x_{i,m}
=
A_{i,m} x_{i,m} -B_i v_{i,h},
\label{num:desired_closed_loop_ref_system}
\end{align}
with hedging signals $v_{\phi,h}=g_\phi\delta_h$, $v_{\theta,h}=g_\theta \cos(\phi)\delta_h$, $v_{r,h}=g_r\delta_h$, we assume $K_i$ to be chosen so that $A_{i,m}=A_i+B_iK_i$ is Hurwitz. We refer to (\ref{num:desired_closed_loop_ref_system}) as the reference model describing the desired tracking error dynamics. However, by substituting the considered baseline controllers into (\ref{num:error_dynamics_systems}), the true error dynamics are
\begin{subequations}\label{num:error_dynamics_systems_with_baseline}
\begin{align}
\dot x_\phi
&= 
A_{\phi,m} x_\phi
+
B_\phi
( \mu_{\phi,\mathrm{ad}} 
+ \Delta f_{\phi} + \Delta g_{\phi} \delta
-g_{\phi} \delta_h
),\\
\dot x_\theta
&= 
A_{\theta,m} x_\theta
+
B_\theta 
(\mu_{\theta,\mathrm{ad}} \notag\\
&  \qquad \qquad 
+
\cos(\phi) 
(
\Delta f_{\theta} 
+ \Delta g_{\theta}  \delta 
-g_{\theta} \delta_h
)
),\\
\dot x_r
&= 
A_{r,m} x_r
+
B_r
( \mu_{r,\mathrm{ad}} 
+ \Delta f_{r} + \Delta g_{r} \delta
-g_{r} \delta_h).
\end{align}
\end{subequations}
Based on the the reference model (\ref{num:desired_closed_loop_ref_system}) and the true error dynamics (\ref{num:error_dynamics_systems_with_baseline}), we define the model matching errors $e_{i,m}=x_i-x_{i,m}$ for $i\in\{\phi,\theta, r\}$ with dynamics
\begin{subequations}\label{num:matching_error_dynamics}
\begin{align}
\dot e_{\phi,m}
&= 
A_{\phi,m} e_{\phi,m}
+
B_\phi
( \mu_{\phi,\mathrm{ad}} 
+ \Delta f_{\phi} + \Delta g_{\phi} \delta
),\\
\dot e_{\theta,m}
&= 
A_{\theta,m} e_{\theta,m}
+
B_\theta 
(\mu_{\theta,\mathrm{ad}} \notag\\
&  \qquad \qquad 
+
\cos(\phi) 
(
\Delta f_{\theta} 
+ \Delta g_{\theta}  \delta 
)
)
,\\
\dot e_{r,m}
&= 
A_{r,m} e_{r,m}
+
B_r
( \mu_{r,\mathrm{ad}} 
+ \Delta f_{r} + \Delta g_{r} \delta).
\end{align}
\end{subequations}

The main goal of the adaptive augmentation is to drive $e_{i,m}$ as close as possible to zero by suitable choice of the remaining inputs $\mu_{i,\mathrm{ad}}$ with $i\in\{\phi,\theta, r\}$. Note, by applying PCH in the reference model (\ref{num:desired_closed_loop_ref_system}), the matching error dynamics (\ref{num:matching_error_dynamics}) become invariant of the safety filter action $\delta_h$. Therefore, by driving $e_{i,m}$ to zero the RL-agent will aim to compensate for the matched uncertainties $\Delta f_i$ and $\Delta g_i$ but not for the safety filter action.

To augment the baseline controller in a manner analogous to MRAC, we feedback the Lyapunov-weighted error
\(
s_i = e_{i,m}^{\top} P_i B_i
\)
to the RL-agent with $i\in\{\phi,\theta, r\}$, 
where $P_i$ is obtained from the Lyapunov equation
\(
A_{i,m}^{\top} P_i + P_i A_{i,m} = -Q
\)
with $Q \succ 0$.

\section{RL-Based Adaptive Augmentation}
\label{sec:adaptive_augmentation}
In this section, we develop a novel RL-based adaptive augmentation that is inspired by the uncertainty compensating
structure of MRAC. The objective is to enable the augmentation to learn to compensate for matched uncertainties in the matching error dynamics \eqref{num:matching_error_dynamics}. Specifically, the RL policy learns an adaptation rule for updating the estimates of the uncertainty parameters rather than directly commanding control surface deflections.

For the uncertainty terms defined in \eqref{num:def_elements_f_uncertainty} and \eqref{num:def_elements_g_uncertainty}, consider the corresponding estimates $\Delta\hat f_\omega\in\mathbb R^3$, $\Delta\hat g_\omega\in\mathbb R^{3\times3}$ and $\Delta\hat f_i\in\mathbb R$, $\Delta\hat g_i\in\mathbb R^{1\times3}$ for $i\in\{\phi,\theta,r\}$. 

The policy is evaluated at $t_k=kT_s$, $k\in\mathbb N_0$, with
sampling period $T_s>0$. For continuous-time signals, the subscript \(k\) denotes evaluation at \(t_k\). Control-dependent observations refer to the values available at the time of the policy call, before the new action is issued. The aircraft, baseline controller and reference model dynamics remain
continuous-time.

The matched uncertainties are not directly measured, so the
sampled learning problem is treated as a partially observable
Markov decision process (POMDP)~\cite{KAELBLING199899}. To supply
temporal information without explicitly estimating a belief state,
the policy uses the current observation and three preceding
observations.
Defining the instantaneous observation, without its history, as
\begin{equation}
 \begin{aligned}
 o_k=\operatorname{col}\big(&o_{\phi,k},o_{\theta,k},o_{r,k},\delta_k,\\
       &\|\Delta\hat f_{\omega,k}\|_2,
        \|\Delta\hat g_{\omega,k}\|_F,o_{x,k}\big),
 \end{aligned}
 \label{eq:rl_obs}
\end{equation}
where
\begin{align*}
 o_{i,k}=\operatorname{col}\big(&e_{i,m,k},s_{i,k},\mu_{i,k},
       \mu_{i,\mathrm{ad},k},\\
       &\Delta\hat f_{i,k},\Delta\hat g_{i,k}^{\top}\big),\\
 o_{x,k}=\operatorname{col}\big(&\phi_k,\theta_k,p_k,q_k,r_k,
             \alpha_k,\beta_k,\gamma_k,v_{\mathrm{TAS},k}\big)
\end{align*}
with $i\in\{\phi,\theta,r\}$. Here $s_{i,k}=e_{i,m,k}^{\top}P_iB_i$ is the sampled
Lyapunov-weighted error already defined in Section~\ref{sec:baseline_controller}.
Both the sampled matching error $e_{i,m,k}$ as well as $s_{i,k}$ quantify the deviation from the reference model. The sampled control inputs $\delta_k$, $\mu_{i,k}$, $\mu_{i,\mathrm{ad},k}$ and the estimated uncertainty parameters $\|\Delta\hat f_{\omega,k}\|_2$, $\|\Delta\hat g_{\omega,k}\|_F$, $\Delta\hat f_{i,k}$, $\Delta\hat g_{i,k}^{\top}$ describe
the current compensation, while $o_{x,k}$ supplies flight-condition
information.

With an initialized history buffer, the input to the policy and the resulting action are given by
\begin{equation}
 \begin{aligned}
  \bar o_k&=\operatorname{col}(o_k,o_{k-1},o_{k-2},o_{k-3}),\\
  a_k&\sim\pi(\cdot\mid\bar o_k).
 \end{aligned}
 \label{eq:rl_stacked_obs}
\end{equation}
The actions $a_k$ are structured as
\begin{equation}
 \begin{aligned}
 a_k=\operatorname{col}\big(&a_{\phi,f,k},a_{\phi,g,k}^{\top},
          a_{\theta,f,k},\\
          &a_{\theta,g,k}^{\top},a_{r,f,k},a_{r,g,k}^{\top}\big)
          \in\mathcal A\subset\mathbb R^{12},
 \end{aligned}
 \label{eq:rl_action}
\end{equation}
where $a_{i,f,k}\in\mathbb R$ and
$a_{i,g,k}\in\mathbb R^{1\times3}$ are later applied as adaptation rates for the uncertainty estimates. The network inputs and outputs are normalized using fixed scaling factors, with the same scaling applied during training and deployment. 

We model the actions with a
zero-order hold,
\begin{equation}
 a(t)=a_k,\qquad t\in[t_k,t_{k+1}),
 \label{eq:mj-policy-hold}
\end{equation}
and for $i\in\{\phi,\theta,r\}$ we retain the continuous-time adaptation 
{\color{black}
\begin{subequations}
 \label{eq:rl_action_mapped2rates}
 \begin{align}
  \Delta\dot{\hat f}_i(t)
   &=\mathrm{Proj}\big(\Delta\hat f_i(t),a_{i,f}(t)\big),\\
  \Delta\dot{\hat g}_i(t)
   &=\mathrm{Proj}\big(\Delta\hat g_i(t),a_{i,g}(t)\big)
 \end{align}
\end{subequations}}
of the estimated parameters. The projection operator $\mathrm{Proj}(\cdot,\cdot)$ is defined as in standard MRAC ~\cite[Ch.~11]{eugene2013robust} and restricts the parameters to a prescribed bounded convex set, avoiding parameter drift. 


The adaptive augmentation is formulated to compensate for the uncertainties in \eqref{num:matching_error_dynamics} as
{\color{black}
\begin{equation}
 \mu_{i,\mathrm{ad}}(t)=
 -\begin{bmatrix}\Delta\hat f_i(t)&\Delta\hat g_i(t)\end{bmatrix}
   \Phi_i(t), 
   \quad i\in\{\phi,\theta,r\}
 \label{eq:mu_ad_law}
\end{equation}}
with regressors
\[
 \begin{aligned}
 \Phi_\phi(t)=\Phi_r(t)&=\operatorname{col}(1,\delta(t)),\\
 \Phi_\theta(t)&=\cos\phi(t)\operatorname{col}(1,\delta(t)).
 \end{aligned}
\]
It is noted that here $\delta$ is the applied, safety-filtered input, as discussed in
 Section~\ref{sec:baseline_controller}, and not the pre-filter command $\delta_{\mathrm{nom}}$.

The complete sampled matching error vector is
\begin{equation}
 e_{m,k}=\operatorname{col}
        (e_{\phi,m,k},e_{\theta,m,k},e_{r,m,k})\in\mathbb R^5,
 \label{eq:mj-matching-vector}
\end{equation}
where $e_{\phi,m,k},e_{\theta,m,k}\in\mathbb R^2$ and
$e_{r,m,k}\in\mathbb R$. Thus
\[
 \|e_{m,k}\|_2^2=\|e_{\phi,m,k}\|_2^2+
                 \|e_{\theta,m,k}\|_2^2+|e_{r,m,k}|^2.
\]
Using the sampled signals and an initialized previous action,
the stage reward is
\begin{align}
 r_k={}&-r_{m,1}\big(1-e^{-r_{m,2}\|e_{m,k}\|_2^2}\big)
          -r_{m,3}\|e_{m,k}\|_2\nonumber\\
       &-r_f\|\Delta\hat f_{\omega,k}\|_2^2
          -r_g\|\Delta\hat g_{\omega,k}\|_F^2\nonumber\\
       &-r_{\mathrm{ad}}\|\mu_{\mathrm{ad},k}\|_2^2
          -r_a\|a_k-a_{k-1}\|_2^2,
 \label{eq:mj-original-reward}
\end{align}
with weights $r_{m,1},r_{m,2},r_{m,3},r_f,r_g,
 r_{\mathrm{ad}},r_a\in\mathbb{R}_{\ge0}$.
The bounded exponential and linear-norm terms penalize the
matching error. The remaining terms penalize parameter magnitude,
augmentation magnitude, and changes in the adaptation rates.
The objective of the RL-based adaptive augmentation is to learn the policy that maximizes the expected cumulative discounted reward, leading to the convergence of $e_{i,m}$. It is noted that the stacked observations introduced in (\ref{eq:rl_stacked_obs}) should help the agent to understand the temporal dynamics of the environment introduced by different perturbations, that are not directly observable.
As such, the training aims to maximize the expected discounted return
\begin{equation}
 J(\pi)=\mathbb E_\pi\!\left[
       \sum_{k=0}^{\infty}\gamma_{\mathrm{RL}}^{\,k}r_k\right],
 \qquad 0\leq\gamma_{\mathrm{RL}}<1.
 \label{eq:mj-discounted-objective}
\end{equation}
Training details and numerical
weight values are detailed in Section~\ref{sec:num_example}.

\section{Safety Filter Design}
\label{sec:safety_filter}
The control design considered so far does not explicitly account for operational flight-envelope constraints and may therefore generate control inputs that drive the UAV toward unsafe regions of the state space. To address this issue, a safety filter of the form (\ref{eq:mj-joint-cbf-qp}) is introduced to modify the nominal control input when necessary to satisfy the considered constraints.


To reduce the conservatism induced by system uncertainties, we employ the disturbance observer (DOB) originally proposed by \cite{chen2000nonlinear} and recently applied in the context of CBFs, e.g., by \cite{dacs2022robust}. For the system (\ref{num:rot_dyanmics_perturbated_underbrace}) with lumped uncertainty $\Delta(x,\delta)=\Delta f_\omega(x) + \Delta g_\omega(x)\delta$ this yields the observer dynamics
\begin{subequations}
\begin{align}
	\dot z
	&
	=
	-\Lambda (
	z
	+
	\Lambda
	\omega
	+
	f_\omega(\omega)
	+
	g_\omega
	\delta
	),\\
	\hat \Delta
	&=
	z + \Lambda \omega
\end{align}
\label{num:DOB_equations}
\end{subequations}
where $\Lambda\succ0$ is a design matrix and $\hat \Delta$ is an estimation of $\Delta(x,\delta)$. Considering $\Delta(x,\delta)=\hat \Delta+\tilde \Delta(x,\delta)$ with $\tilde \Delta(x,\delta) = \Delta(x,\delta)-\hat \Delta$, we can write the system dynamics as
\begin{align}
\dot{\omega}
&=
f_{\omega}(x)
+
g_{\omega}(x)\delta
+
\hat \Delta
+
\tilde \Delta(x,\delta),
\label{num:rot_dyanmics_perturbated_w_observer}
\end{align}
where the only uncertainty is $\tilde \Delta(x,\delta)$ that can be bounded based on Proposition \ref{prop:DOB_error_bound}. In the following we drop the dependency of $\Delta(x,\delta)$ on $x$ and $\delta$ for notational convenience. 
\begin{proposition}
\label{prop:DOB_error_bound}
The disturbance estimation error $\tilde \Delta = \Delta-\hat \Delta$ decays exponentially towards the ultimate bound
\(
\mathcal{B} = \{\tilde \Delta\colon \|\tilde \Delta\|
\le
\frac{\dot \Delta_{\mathrm{max}}}
{\lambda_{\mathrm{min}}(\Lambda)}
\}
\) 
with $\dot \Delta_{\mathrm{max}}$ from (\ref{eq:_bounded_uncertainty_lumped_aircraft}) and an exponential decay rate of $\lambda_{\mathrm{min}}(\Lambda)$.
\end{proposition}
\begin{proof}
With $\tilde \Delta = \Delta-\hat \Delta$ and (\ref{num:DOB_equations}) it follows
\(
	\dot{\tilde \Delta}
	=
	\dot \Delta 	-\Lambda \tilde \Delta.
\)
Considering the Lyapunov function candidate $V=\frac{1}{2}\tilde \Delta^{\top}\tilde \Delta$ and substituting 
\(
\dot{\tilde \Delta}
	=
	\dot \Delta 	-\Lambda \tilde \Delta
\) 
into 
\(
\dot V
	=
	\tilde \Delta^{\top}\dot {\tilde \Delta}
\) after some algebra yields
\begin{align}
	\dot V
	&
	\le
	-
	2
	\lambda_{\mathrm{min}}(\Lambda)V
	+ 
	\sqrt{2}\dot \Delta_{\mathrm{max}}\sqrt{V} .
	\label{num:DOB_prop_result_V}
\end{align}
Applying the substitution $\sqrt{V}=V_s$ with $\dot V = 2\dot V_sV_s$ leads to
\(
	\dot V_s
	\le
	-
	\lambda_{\mathrm{min}}(\Lambda)V_s
	+ 
	\frac{\sqrt{2}\dot \Delta_{\mathrm{max}}}{2}
\)
with the solution
\begin{align}
	V_s(t)
	&
	\le
	e^{-\lambda_{\mathrm{min}}(\Lambda)t}
	V_s(0)
	+
	\frac{\dot \Delta_{\mathrm{max}}}
	{\sqrt{2}\lambda_{\mathrm{min}}(\Lambda)}
(1-e^{-\lambda_{\mathrm{min}}(\Lambda)t}).\notag
\end{align}
Then with $V_s(t)=\frac{1}{\sqrt{2}}\|\tilde\Delta(t)\|$ it follows the claim.
\end{proof}

For the angular rates, we consider the following constraint functions
\begin{align}
    h_i(x) = i_{\mathrm{max}}^2 - (C_i\omega)^2\ge0,
    \quad i\in{p,q,r},
    \label{num:cbf_angeluar_rates_def}
\end{align}
with rate limits $p_{\mathrm{max}}, q_{\mathrm{max}}, r_{\mathrm{max}}\in\mathbb{R}_{>0}$ and \(
C_p
=
\begin{bmatrix}
1 & 0 & 0
\end{bmatrix}
\),
\(
C_q
=
\begin{bmatrix}
0 & 1 & 0
\end{bmatrix}
\),
\(
C_r
=
\begin{bmatrix}
0 & 0 & 1
\end{bmatrix}
\). Then (\ref{num:cbf_angeluar_rates_def}) defines the safe set
\(
    \mathcal{C}_\omega = \{x \in \mathcal{X} : h_p(x) \geq 0 \land h_q(x) \geq 0 \land h_r(x) \geq 0 \}
\).
Consequently, $h_i(x)$ with $i\in{p,q,r}$ is a CBF if
\begin{align}
&
\sup_{\delta\in U}[
\dot h_i(x,\delta)
]
\ge
-\gamma_i h_i(x)
\label{num:CBF_condition_rates}
\end{align}
for all $x\in\mathcal{X}$, where 
\(
\dot h_i(x,\delta)
=
- 2(C_i\omega)C_i
(
f_{\omega}(x)
+
g_{\omega}(x)\delta
+
\hat \Delta
+
\tilde \Delta
)
\)
and
\(
\gamma_i\in\mathbb{R}_{>0}
\)
is a tuning parameter \cite{amesControlBarrierFunctions2019}.
Using the bound 
\(
2|C_i\omega|
(
\frac{\dot \Delta_{\mathrm{max}}}
{\lambda_{\mathrm{min}}(\Lambda)}
+
\epsilon_{\mathrm DOB}
)
\ge
\|2(C_i\omega)C_i
\tilde \Delta\|
\)
from Proposition \ref{prop:DOB_error_bound}, a sufficient condition for (\ref{num:CBF_condition_rates}) to hold is 
\begin{align}
&
\sup_{\delta\in U}[
- 2(C_i\omega)C_i
(
f_{\omega}(x)
+
g_{\omega}(x)\delta
+
\hat \Delta
)
]\ge\notag\\
&
\qquad \qquad \quad  
-\gamma_i h_i(x) 
+
2|C_i\omega|
(
\frac{\dot \Delta_{\mathrm{max}}}
{\lambda_{\mathrm{min}}(\Lambda)}
+
\epsilon_{\mathrm DOB}
),
\label{num:CBF_condition_ub}
\end{align}
where $\epsilon_{\mathrm DOB}\in\mathbb{R}_{\ge0}$ is a tuning parameter to account for the initial transients of the DOB.

For the Euler angles $\phi$ and $\theta$, we consider the following constraint functions
\begin{align}
    h_i(x) = i_{\mathrm{max}}^2 - (C_i\Theta)^2\ge0,
    \quad i\in{\phi,\theta},
    \label{num:cbf_Euler_angles_def}
\end{align}
with limits $\phi_{\mathrm{max}}, \theta_{\mathrm{max}}\in\mathbb{R}_{>0}$ and \(
C_\phi
=
\begin{bmatrix}
1 & 0 & 0
\end{bmatrix}
\),
\(
C_\theta
=
\begin{bmatrix}
0 & 1 & 0
\end{bmatrix}
\). Then (\ref{num:cbf_Euler_angles_def}) defines the safe set
\(
    \mathcal{C}_\Theta = \{x \in \mathcal{X} : h_\phi(x) \geq 0 \land h_\theta(x) \geq 0 \}
\). It follows that $h_i(x)$ with $i\in{\phi,\theta}$ is a HOCBF if
\begin{align}
&
\sup_{\delta\in U}[
\dot \psi_1(x,\delta)
]
\ge
-
\gamma_{i,2}\psi_1(x)
\label{num:HOCBF_condition_Euler}
\end{align}
for all $x\in\mathcal{X}$ with  
\(
\psi_1(x)
=
\dot h_i(x)
+
\gamma_{i,1}h_i(x)
\)
and
\(
\dot \psi_1(x,\delta)
=
\ddot h_i(x,\delta)
+
\gamma_{i,1}\dot h_i(x)
\),
where
\(
\dot h_i(x)
=
- 2(C_i\Theta)C_i\dot\Theta  
\),
\(
\ddot h_i(x,\delta)
=
- 2(C_i\Theta)C_i\Psi
(f_{\omega}(x)
+
g_{\omega}(x)\delta
+
\hat \Delta
+
\tilde \Delta
)
+
c_i(x)
\),
\(
c_i(x)
=
- 2(C_i\dot\Theta)^2
- 2(C_i\Theta)C_i\dot\Psi\omega
\)
and with tuning parameters $\gamma_{i,1},\gamma_{i,2}\in\mathbb{R}_{>0}$ \cite{xiaoControlBarrierFunctions2019}.
Using the bound 
\(
2|C_i\Theta|\|C_i\Psi\|
(
\frac{\dot \Delta_{\mathrm{max}}}
{\lambda_{\mathrm{min}}(\Lambda)}
+
\epsilon_{\mathrm DOB}
)
\ge
\|2(C_i\Theta)C_i\Psi
\tilde \Delta\|
\)
from Proposition \ref{prop:DOB_error_bound}, a sufficient condition for (\ref{num:HOCBF_condition_Euler}) to be satisfied is
\begin{align}
&
\sup_{\delta\in U}[
- 2(C_i\Theta)C_i\Psi
(f_{\omega}(x)
+
g_{\omega}(x)\delta
+
\hat \Delta
)
]\ge 
-
\gamma_{i,2}\psi_1(x)
\notag\\
&
-
\gamma_{i,1}\dot h_i(x)
-
c_i(x)
+
2|C_i\Theta|\|C_i\Psi\|
(
\frac{\dot \Delta_{\mathrm{max}}}
{\lambda_{\mathrm{min}}(\Lambda)}
+
\epsilon_{\mathrm DOB}
)
\label{num:HOCBF_condition_ub},
\end{align}
where $\epsilon_{\mathrm{DOB}}\in\mathbb{R}_{\ge0}$ again accounts for the initial DOB transients.

To satisfy the rate and angle constraints with minimum invasive control action we formulate the QP
\begin{align}
    &\delta(x)=\arg \min_{\delta\in\mathcal{U}}
    \frac{1}{2}
    \|\delta - \delta_\mathrm{nom}(x)\|\\
    & \text{s.t.} \qquad
    A_\mathrm{CBF}(x)\delta\le b_\mathrm{CBF}(x),\notag
\end{align}
where $A_\mathrm{CBF}(x)\in\mathbb{R}^{5\times 3}$ and $b_\mathrm{CBF}(x)\in\mathbb{R}^{5}$ are obtained by rearranging (\ref{num:CBF_condition_ub}) and (\ref{num:HOCBF_condition_ub}) accordingly.

\section{Numerical Example}
\label{sec:num_example}
The RL-based adaptive augmentation combined with the safety filter is evaluated for a representative tracking scenario, which was obtained after training PPO from Stable Baselines 3 for $130M$ steps using standard hyper-parameters.
The outer-loop guidance law generates the reference commands for $\chi_c$, $h_c$, and $v_{\mathrm{TAS, c}}$, with command changes applied at prescribed switching times. The reference commands are given by $\chi_c \in \{(\SI{0.0}{\second}, \SI{0.0}{\degree}), (\SI{30.0}{\second}, \SI{100.0}{\degree}), (\SI{60.0}{\second},\SI{0.0}{\degree})\}$, $h_c \in \{(\SI{0.0}{\second}, \SI{500.0}{\meter}), (\SI{35.0}{\second}, \SI{750.0}{\meter})\}$, and $v_{\mathrm{TAS, c}} \in \{(\SI{0.0}{\second}, \SI{30.0}{\meter\per\second}), (\SI{50.0}{\second}, \SI{40.0}{\meter\per\second})\}$.
The simulation parameters are listed in Table \ref{tab:simulation_params}.
The comparison of resulting closed-loop responses are shown in Fig. \ref{fig:num_example}.
Fig. \ref{fig:attitude_ndi} shows the closed-loop response of the baseline NDI controller without the safety filter. The controller does provide satisfactory tracking performance for $\phi$ and $r$, however, does not provide accurate tracking of the generated $\theta_r$ reference, and drives $\phi$ beyond its prescribed attitude constraint, illustrating the need for an adaptive augmentation to compensate the matched uncertainties in the closed-loop error dynamics to achieve accurate reference tracking, and a safety filter for flight envelope constraint satisfaction.
The RL-based adaptive augmentation with safety filter and PCH is shown in Fig.  \ref{fig:attitude_rl_based_adaptive_augmentation}, illustrating accurate tracking of $\phi_r$, $\theta_r$, and $r_r$ when flight envelope constraints are not violated. At approximately $t = \SI{35}{\second}$ and $t = \SI{65}{\second}$, the outer-loop reference commands are not safe, i.e. would drive the aircraft beyond its prescribed attitude constraints for $\phi$ and $\theta$, and the safety filter intervenes, prioritizing flight envelope protection over accurate reference tracking.
In these intervals, where the reference commands or the RL-based adaptive augmentation would lead to flight envelope constraint violations, the safety filter modifies the nominal control input $\delta_{\mathrm{nom}}$ by the safe input $\delta$ to satisfy flight envelope constraints.
This is shown in Fig. \ref{fig:control_surfaces_e_norm_adaptive_parameters_pch}, which shows the nominal control input $\delta_{\mathrm{nom}}$, the safe control input $\delta$, the deviation between nominal control input and safe control input $\delta_{\mathrm{nom}} - \delta$, the model-matching error norm $\lVert e_m \rVert$, and the adaptive parameters from the RL-based adaptive augmentation. Note that $\delta = \delta_{\mathrm{nom}}$ if no flight envelope constraints are violated. This illustrates that the safety filter only intervenes when flight envelope constraints are violated, thereby modifying the RL-based adaptive augmentation in an minimal invasive way, while PCH ensures that the RL-based adaptive augmentation does not compensate for interventions introduced by the safety filter, making the model-matching error dynamics invariant of the input discrepancy introduced by the safety filter. Fig. \ref{fig:control_surfaces_e_norm_adaptive_parameters_without_pch} illustrates the effects of omitting PCH. Since the safety filter induced control modifications are not reflected in the reference model, the model-matching error dynamics are no longer invariant from safety filter input discrepancy, resulting in an increased RL-based adaptive augmentation, causing the adaptive parameters to increase until they reach the projection boundary, potentially leading to degraded closed-loop stability.

\begin{figure*}[t]
    \centering

    \begin{subfigure}[t]{0.48\textwidth}
        \centering
        \includegraphics[width=\linewidth]{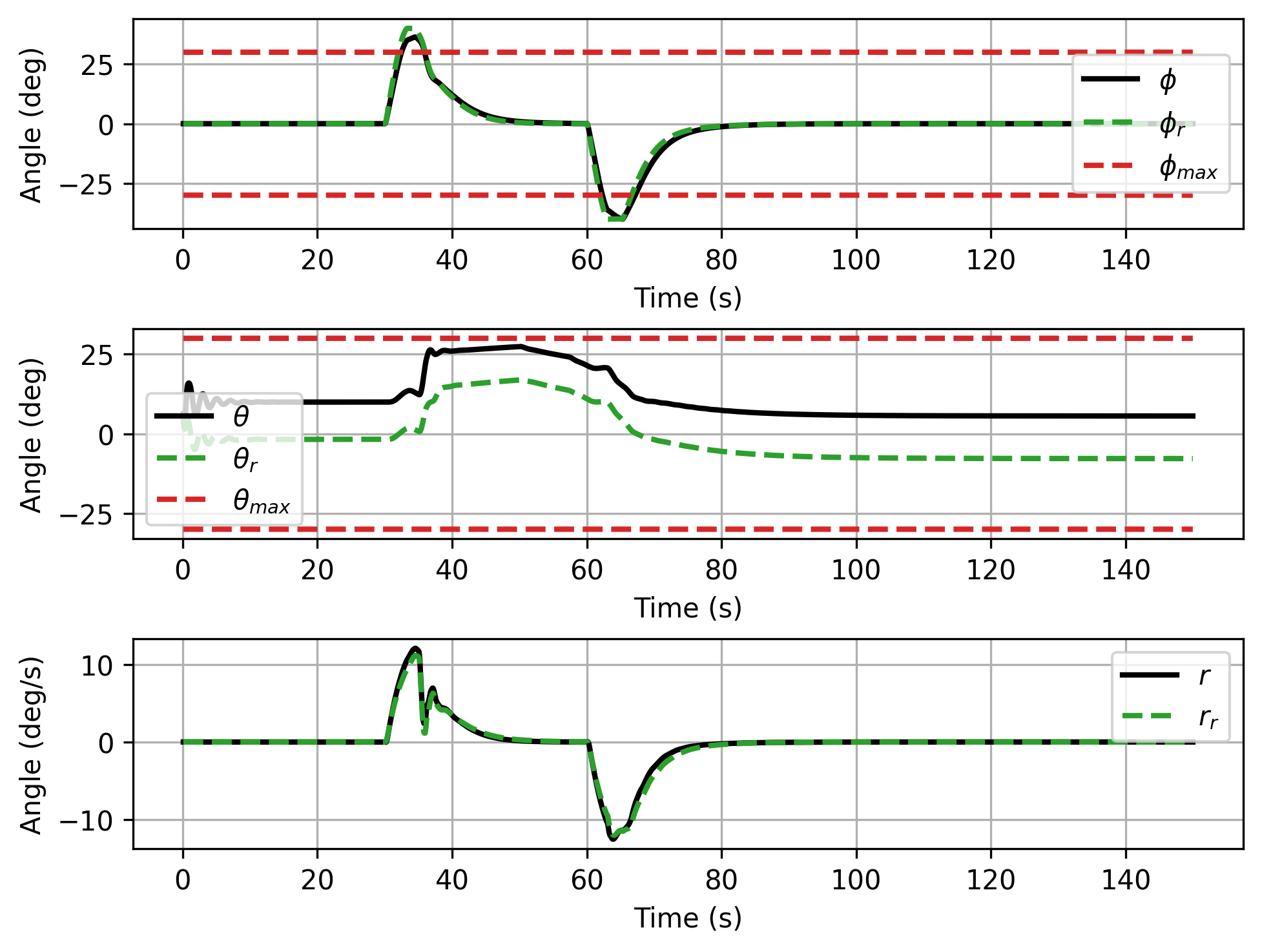}
        \caption{Outer-loop reference commands and resulting closed-loop response from baseline NDI controller}
        \label{fig:attitude_ndi}
    \end{subfigure}
    \hfill
    \begin{subfigure}[t]{0.48\textwidth}
        \centering
        \includegraphics[width=\linewidth]{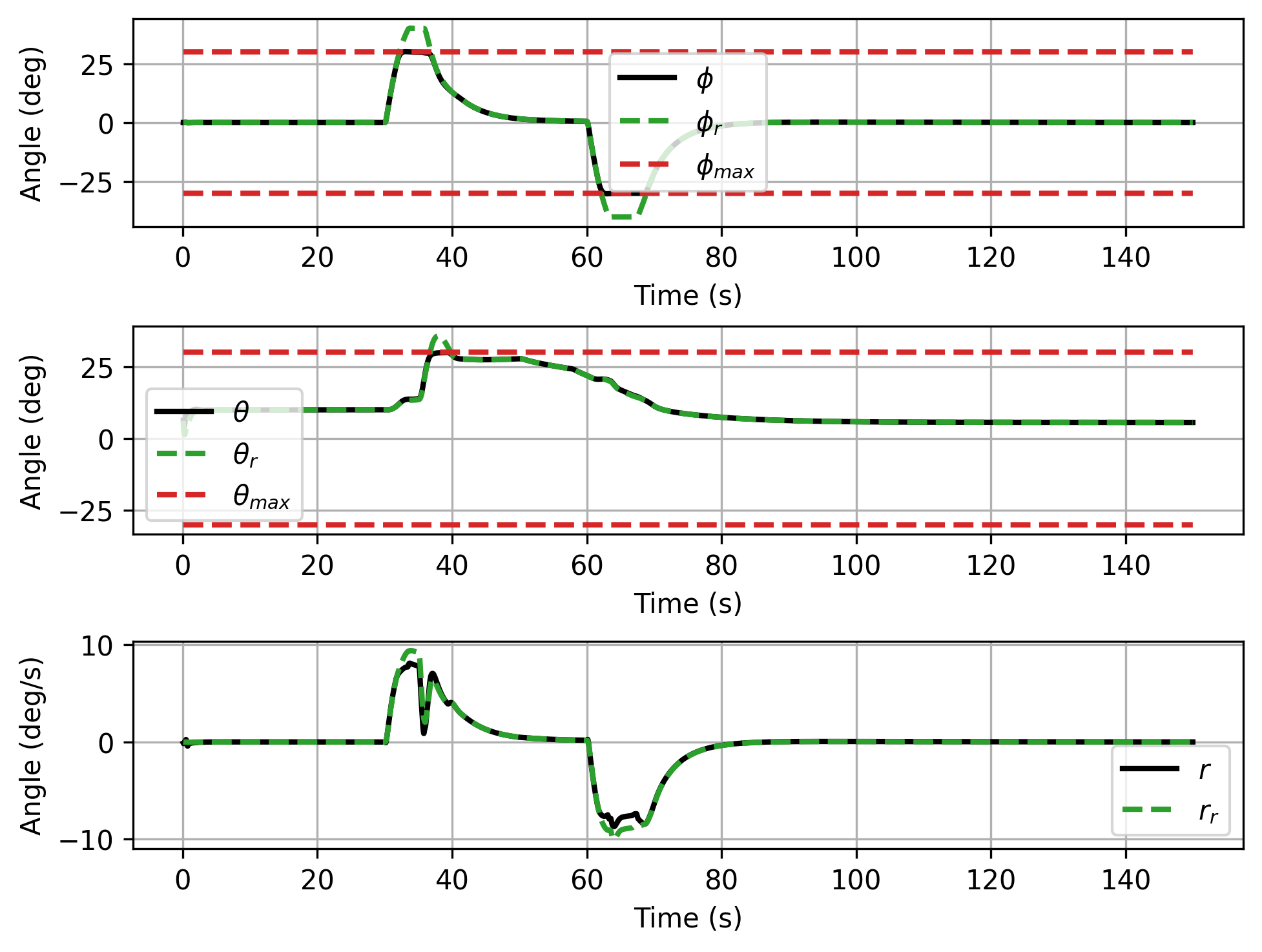}
        \caption{Outer-loop reference commands and resulting closed-loop response from RL-based adaptive augmentation with safety filter and PCH}
        \label{fig:attitude_rl_based_adaptive_augmentation}
    \end{subfigure}

    \vspace{0.5em}

     \begin{subfigure}[t]{0.48\textwidth}
        \centering
        \includegraphics[width=\linewidth]{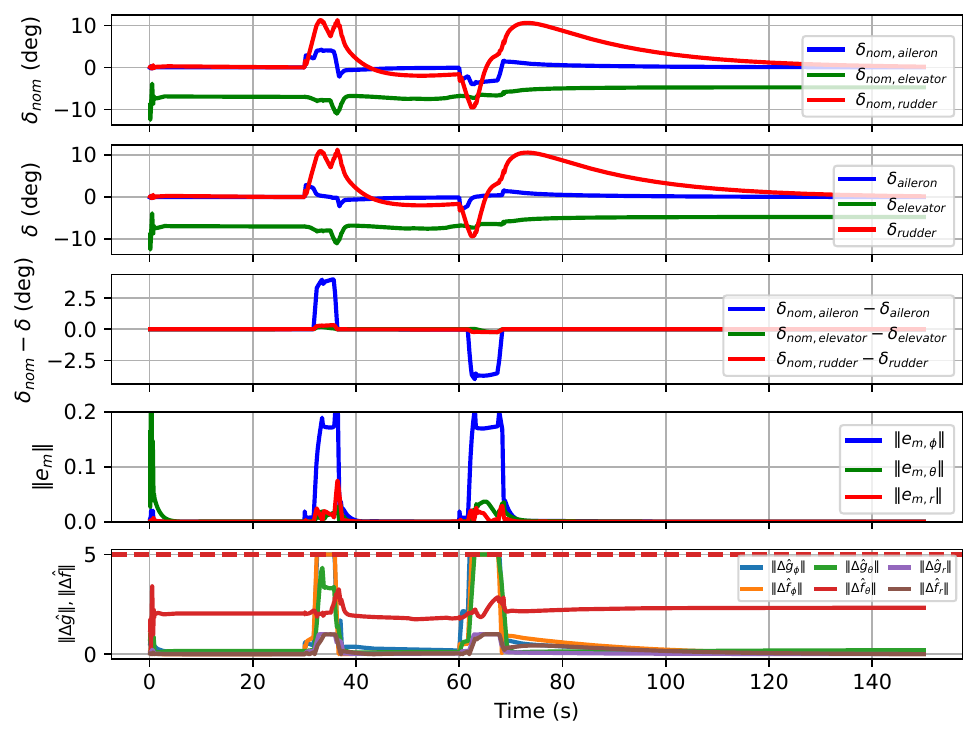}
        \caption{Nominal and safety-filtered control inputs, matching error norm and adaptive parameters from RL-based adaptive augmentation without PCH}
        \label{fig:control_surfaces_e_norm_adaptive_parameters_without_pch}
    \end{subfigure}
    \hfill
    \begin{subfigure}[t]{0.48\textwidth}
        \centering
        \includegraphics[width=\linewidth]{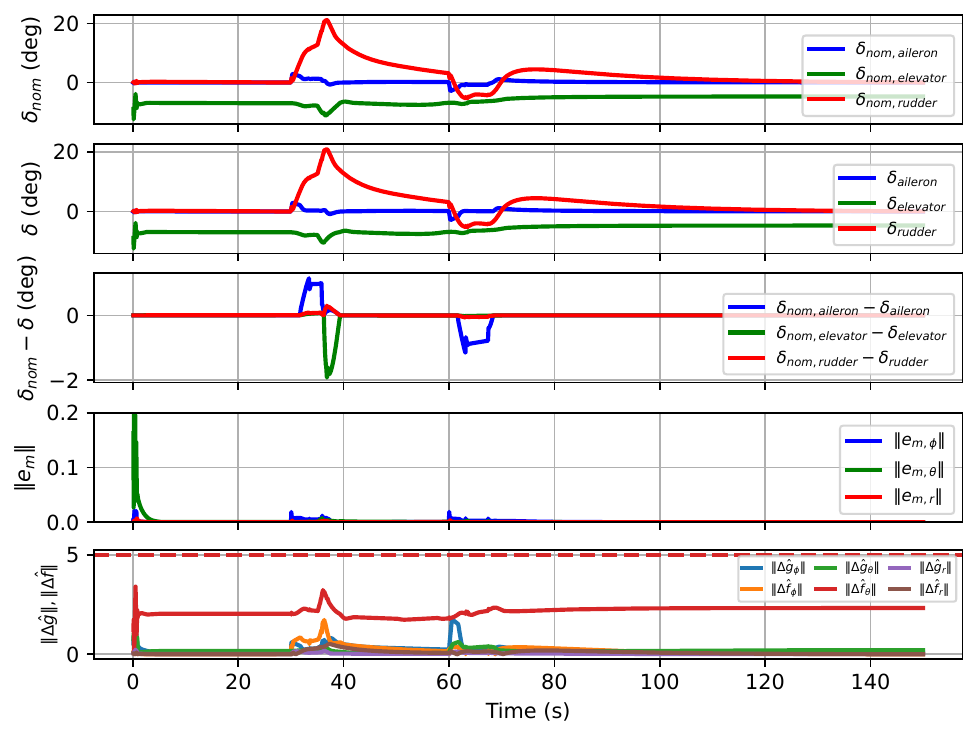}
        \caption{Nominal and safety-filtered control inputs, matching error norm and adaptive parameters from RL-based adaptive augmentation with PCH}
        \label{fig:control_surfaces_e_norm_adaptive_parameters_pch}
    \end{subfigure}
    \caption{Closed-loop comparison of NDI and safety-filtered RL-based adaptive control with and without PCH}
    \label{fig:num_example}
\end{figure*}

\begin{table}[H]
    \centering
    \caption{Simulation parameters}
    \begin{tabular}{lclc}
        \toprule
        Parameter & Value & Parameter & Value \\
        \midrule

        \multicolumn{4}{c}{\textit{Simulation parameters}} \\
        \midrule
        $p_{\max}$        & \SI{100.0}{\degree\per\second} &
        $q_{\max}$        & \SI{60.0}{\degree\per\second} \\

        $r_{\max}$        & \SI{50.0}{\degree\per\second} &
        $\phi_{\max}$     & \SI{30.0}{\degree} \\

        $\theta_{\max}$   & \SI{30.0}{\degree} &
        $\gamma_{p,q,r}$  & 5.0 \\

        $\gamma_{\phi, \theta, 1}$        & 3.0 &
        $\gamma_{\phi, \theta, 2}$        & 5.0 \\

        $\delta_{\max}$   &
        $[\SI{25.0}{\degree},
          \SI{30.0}{\degree},
          \SI{90.0}{\degree}]$ &
        $\Delta t$        & \SI{0.02}{\second} \\
        $N_{steps} $ & \SI{7500}{} \\
        \midrule
        \multicolumn{4}{c}{\textit{Reward function parameters}} \\
        \midrule
        $r_{m,1}$ & 10.0 &
        $r_{m,2}$ & 0.5 \\

        $r_{m,3}$ & 10.0 &
        $r_d$     & 0.15 \\

        $r_C$     & 0.15 &
        $r_{ad}$  & 0.05 \\

        $r_c$     & 0.001 &
        & \\

        \bottomrule
    \end{tabular}
    \label{tab:simulation_params}
\end{table}


\section{Conclusions}
\label{sec:conclusion}

This paper presents a safe RL-based adaptive augmentation control scheme, in which RL is used to learn an adaptation law for compensating the effects of matched uncertainties. A safety filter is employed to keep the states within a predefined safe set, while pseudo control hedging is formulated to avoid interactions between the RL augmentation and the safety filter. The proposed concept is applied to fixed-wing UAV attitude control under uncertainty. The results demonstrate effective uncertainty compensation and successful avoidance of adverse interactions with the safety filter while maintaining the prescribed flight-envelope constraints.








\bibliographystyle{IEEEtran}
\bibliography{mybibfile}

@inproceedings{dacs2022robust,
  title={Robust safe control synthesis with disturbance observer-based control barrier functions},
  author={Da{\c{s}}, Ersin and Murray, Richard M},
  booktitle={2022 IEEE 61st conference on decision and control (CDC)},
  pages={5566--5573},
  year={2022},
  organization={IEEE}
}

@inproceedings{zhao2023stable,
  title={Stable and safe reinforcement learning via a barrier-lyapunov actor-critic approach},
  author={Zhao, Liqun and Gatsis, Konstantinos and Papachristodoulou, Antonis},
  booktitle={2023 62nd IEEE Conference on Decision and Control (CDC)},
  pages={1320--1325},
  year={2023},
  organization={IEEE}
}

@article{marvi2021safe,
  title={Safe reinforcement learning: A control barrier function optimization approach},
  author={Marvi, Zahra and Kiumarsi, Bahare},
  journal={International Journal of Robust and Nonlinear Control},
  volume={31},
  number={6},
  pages={1923--1940},
  year={2021},
  publisher={Wiley Online Library}
}

@article{Markgraf2026,
  title = {Safe Reinforcement Learning using Action Projection: Safeguard the Policy or the Environment?},
  author = {Markgraf, Hannah and Sawant, Shambhuraj and Krasowski, Hanna and Schäfer, Lukas and Gros, Sebastien and Althoff, Matthias},
  year = {2026},
  journal = {Transactions on Machine Learning Research},
}

@article{emam2022safe,
  title={Safe reinforcement learning using robust control barrier functions},
  author={Emam, Yousef and Notomista, Gennaro and Glotfelter, Paul and Kira, Zsolt and Egerstedt, Magnus},
  journal={IEEE Robotics and Automation Letters},
  volume={10},
  number={3},
  pages={2886--2893},
  year={2022},
  publisher={IEEE}
}

@inproceedings{cheng2019end,
  title={End-to-end safe reinforcement learning through barrier functions for safety-critical continuous control tasks},
  author={Cheng, Richard and Orosz, G{\'a}bor and Murray, Richard M and Burdick, Joel W},
  booktitle={Proceedings of the AAAI conference on artificial intelligence},
  volume={33},
  number={01},
  pages={3387--3395},
  year={2019}
}

@article{goel2024composite,
  title={Composite adaptive control for time-varying systems with dual adaptation},
  author={Goel, Raghavv and Roy, Sayan Basu},
  journal={IEEE Transactions on Automatic Control},
  volume={70},
  number={1},
  pages={487--494},
  year={2024},
  publisher={IEEE}
}

@INPROCEEDINGS{KannanACCMRACRL,
  author={Kannan, Harinee and Patnaik, Karishma and Zhang, Wenlong},
  booktitle={2025 American Control Conference (ACC)}, 
  title={Reinforcement Learning-based Hover Control of a Quadrotor with Model Reference Adaptation}, 
  year={2025},
  volume={},
  number={},
  pages={438-443},
  doi={10.23919/ACC63710.2025.11107552}}

@INPROCEEDINGS{Guha2021MRACRL,
  author={Guha, Anubhav and Annaswamy, Anuradha M.},
  booktitle={2021 60th IEEE Conference on Decision and Control (CDC)}, 
  title={Online Policies for Real-Time Control Using MRAC-RL}, 
  year={2021},
  volume={},
  number={},
  pages={1808-1813},
  doi={10.1109/CDC45484.2021.9683641}}

@ARTICLE{Borghesi2026MRARL,
  author={Borghesi, Marco and Bosso, Alessandro and Notarstefano, Giuseppe},
  journal={IEEE Transactions on Automatic Control}, 
  title={{MR-ARL}: Model Reference Adaptive Reinforcement Learning for Robustly Stable On-Policy Data-Driven {LQR}}, 
  year={2026},
  volume={71},
  number={2},
  pages={1129-1144},
  doi={10.1109/TAC.2025.3611155}}

@ARTICLE{Annaswamy2023RLACTAC,
  author={Annaswamy, Anuradha M. and Guha, Anubhav and Cui, Yingnan and Tang, Sunbochen and Fisher, Peter A. and Gaudio, Joseph E.},
  journal={IEEE Transactions on Automatic Control}, 
  title={Integration of Adaptive Control and Reinforcement Learning for Real-Time Control and Learning}, 
  year={2023},
  volume={68},
  number={12},
  pages={7740-7755},
  doi={10.1109/TAC.2023.3290037}}

@ARTICLE{Cheng2022L1RL,
  author={Cheng, Yikun and Zhao, Pan and Wang, Fanxin and Block, Daniel J. and Hovakimyan, Naira},
  journal={IEEE Robotics and Automation Letters}, 
  title={Improving the Robustness of Reinforcement Learning Policies With ${\mathcal {L}_{1}}$ Adaptive Control}, 
  year={2022},
  volume={7},
  number={3},
  pages={6574-6581},
  doi={10.1109/LRA.2022.3169309}}

@article{margolis2024rapid,
  title={Rapid locomotion via reinforcement learning},
  author={Margolis, Gabriel B and Yang, Ge and Paigwar, Kartik and Chen, Tao and Agrawal, Pulkit},
  journal={The International Journal of Robotics Research},
  volume={43},
  number={4},
  pages={572--587},
  year={2024},
  publisher={SAGE Publications Sage UK: London, England}
}

@inproceedings{chen2022understanding,
  title     = {Understanding Domain Randomization for Sim-to-Real Transfer},
  author    = {Chen, Xiaoyu and Hu, Jiachen and Jin, Chi and Li, Lihong and Wang, Liwei},
  booktitle = {International Conference on Learning Representations (ICLR)},
  year      = {2022}
}

@article{wada2022sim,
  title={Sim-to-real transfer for fixed-wing uncrewed aerial vehicle: pitch control by high-fidelity modelling and domain randomization},
  author={Wada, Daichi and Araujo-Estrada, Sergio and Windsor, Shane},
  journal={IEEE Robotics and Automation Letters},
  volume={7},
  number={4},
  pages={11735--11742},
  year={2022},
  publisher={IEEE}
}

@inproceedings{konatala2021reinforcement,
  title={Reinforcement learning based online adaptive flight control for the {Cessna} {Citation} {II} ({PH-LAB}) aircraft},
  author={Konatala, Ramesh and Van Kampen, Erik-Jan and Looye, Gertjan},
  booktitle={AIAA Scitech 2021 Forum},
  pages={0883},
  year={2021}
}

@article{konatala2024flight,
  title={Flight testing reinforcement-learning-based online adaptive flight control laws on {CS-25-Class} aircraft},
  author={Konatala, Ramesh and Milz, Daniel and Weiser, Christian and Looye, Gertjan and van Kampen, E},
  journal={Journal of Guidance, Control, and Dynamics},
  volume={47},
  number={11},
  pages={2460--2467},
  year={2024},
  publisher={American Institute of Aeronautics and Astronautics}
}

@inproceedings{dias2019intelligent,
  title={Intelligent nonlinear adaptive flight control using incremental approximate dynamic programming},
  author={Dias, Pedro Miguel and Zhou, Ye and Van Kampen, Erik-Jan},
  booktitle={AIAA Scitech 2019 Forum},
  pages={2339},
  year={2019}
}

@article{bohn2023data,
  title={Data-efficient deep reinforcement learning for attitude control of fixed-wing {UAVs}: Field experiments},
  author={B{\o}hn, Eivind and Coates, Erlend M and Reinhardt, Dirk and Johansen, Tor Arne},
  journal={IEEE Transactions on Neural Networks and Learning Systems},
  volume={35},
  number={3},
  pages={3168--3180},
  year={2023},
  publisher={IEEE}
}

@inproceedings{dally2022soft,
  title={Soft actor-critic deep reinforcement learning for fault tolerant flight control},
  author={Dally, Killian and Van Kampen, Erik-Jan},
  booktitle={AIAA SciTech 2022 Forum},
  pages={2078},
  year={2022}
}

@article{de2023deep,
  title={A deep reinforcement learning control approach for high-performance aircraft},
  author={De Marco, Agostino and D’Onza, Paolo Maria and Manfredi, Sabato},
  journal={Nonlinear Dynamics},
  volume={111},
  number={18},
  pages={17037--17077},
  year={2023},
  publisher={Springer}
}

@inproceedings{chowdhury2024unified,
  title={A unified inner-outer loop reinforcement learning flight controller for fixed-wing aircraft},
  author={Chowdhury, Mozammal and Keshmiri, Shawn},
  booktitle={2024 International Conference on Unmanned Aircraft Systems (ICUAS)},
  pages={556--563},
  year={2024},
  organization={IEEE}
}

@inproceedings{bohn2019deep,
  title={Deep reinforcement learning attitude control of fixed-wing {UAVs} using proximal policy optimization},
  author={B{\o}hn, Eivind and Coates, Erlend M and Moe, Signe and Johansen, Tor Ame},
  booktitle={2019 international conference on unmanned aircraft systems (ICUAS)},
  pages={523--533},
  year={2019},
  organization={IEEE}
}

@article{chowdhury2024interchangeable,
  title={Interchangeable reinforcement-learning flight controller for fixed-wing {UASs}},
  author={Chowdhury, Mozammal and Keshmiri, Shawn},
  journal={IEEE Transactions on Aerospace and Electronic Systems},
  volume={60},
  number={2},
  pages={2305--2318},
  year={2024},
  publisher={IEEE}
}

@article{shukla2024reinforcement,
  title={Reinforcement learning-based evolving flight controller for fixed-wing uncrewed aircraft},
  author={Shukla, Daksh and Benyamen, Hady and Keshmiri, Shawn and Beckage, Nicole M},
  journal={IEEE Transactions on Control Systems Technology},
  volume={33},
  number={3},
  pages={872--886},
  year={2024},
  publisher={IEEE}
}

@inproceedings{marquis2026adversarial,
  title={Adversarial reinforcement learning for robust control of fixed-wing aircraft under model uncertainty},
  author={Marquis, Dennis J and Wilhelm, Blake and Muniraj, Devaprakash and Farhood, Mazen},
  booktitle={2026 American Control Conference (ACC)},
  pages={1158--1164},
  year={2026},
  organization={IEEE}
}

@article{richter2024review,
  title={A review of reinforcement learning for fixed-wing aircraft control tasks},
  author={Richter, David J and Calix, Ricardo A and Kim, Kyungbaek},
  journal={IEEE Access},
  volume={12},
  pages={103026--103048},
  year={2024},
  publisher={IEEE}
}

@inproceedings{johnson2000pseudo,
	title={Pseudo-control hedging: A new method for adaptive control},
	author={Johnson, Eric N and Calise, Anthony J},
	booktitle={Advances in navigation guidance and control technology workshop},
	pages={1--2},
	year={2000},
	organization={Alabama, USA Alabama, USA}
}

@article{chen2000nonlinear,
	title={A nonlinear disturbance observer for robotic manipulators},
	author={Chen, Wen-Hua and Ballance, Donald J and Gawthrop, Peter J and O'Reilly, John},
	journal={IEEE Transactions on industrial Electronics},
	volume={47},
	number={4},
	pages={932--938},
	year={2000},
	publisher={IEEE}
}

@article{eugene2013robust,
	title={Robust and adaptive control with aerospace applications},
	author={Eugene, Lavretsky and Kevin, Wise and Howe, D},
	journal={England: Springer-Verlag London},
	year={2013}
}

@inproceedings{amesControlBarrierFunctions2019,
  title = {Control {{Barrier Functions}}: {{Theory}} and {{Applications}}},
  shorttitle = {Control {{Barrier Functions}}},
  booktitle = {2019 18th {{European Control Conference}} ({{ECC}})},
  author = {Ames, Aaron D. and Coogan, Samuel and Egerstedt, Magnus and Notomista, Gennaro and Sreenath, Koushil and Tabuada, Paulo},
  date = {2019-06},
  pages = {3420--3431},
  publisher = {IEEE},
  location = {Naples, Italy},
  doi = {10.23919/ECC.2019.8796030},
  url = {https://ieeexplore.ieee.org/document/8796030/},
  urldate = {2026-08-13},
  eventtitle = {2019 18th {{European Control Conference}} ({{ECC}})},
  isbn = {978-3-907144-00-8},
  langid = {english}
}

@article{KAELBLING199899,
  title = {Planning and Acting in Partially Observable Stochastic Domains},
  author = {Kaelbling, Leslie Pack and Littman, Michael L. and Cassandra, Anthony R.},
  year = 1998,
  journal = {Artificial Intelligence},
  volume = {101},
  number = {1},
  pages = {99--134},
  issn = {0004-3702},
  doi = {10.1016/S0004-3702(98)00023-X}
}

@inproceedings{xiaoControlBarrierFunctions2019,
  title = {Control {{Barrier Functions}} for {{Systems}} with {{High Relative Degree}}},
  booktitle = {2019 {{IEEE}} 58th {{Conference}} on {{Decision}} and {{Control}} ({{CDC}})},
  author = {Xiao, Wei and Belta, Calin},
  year = 2019,
  month = dec,
  pages = {474--479},
  publisher = {IEEE},
  address = {Nice, France},
  doi = {10.1109/CDC40024.2019.9029455},
  urldate = {2026-08-26},
  copyright = {https://ieeexplore.ieee.org/Xplorehelp/downloads/license-information/IEEE.html},
  isbn = {978-1-7281-1398-2},
  langid = {english}
}

\end{document}